\documentclass[11pt]{article}
\usepackage{graphicx,xspace,amsmath,amssymb,fullpage,amsthm,booktabs, url}
\usepackage{algorithmicx,algorithm}
\usepackage[dvipsnames,usenames]{color}
\usepackage[noend]{algpseudocode}

\algnewcommand\algorithmicinput{\textbf{INPUT:}}
\algnewcommand\INPUT{\item[\algorithmicinput]}
\algnewcommand\algorithmicoutput{\textbf{OUTPUT:}}
\algnewcommand\OUTPUT{\item[\algorithmicoutput]}

\newcommand{\OCT}{\textsc{Odd Cycle Transversal}\xspace}

\newcommand{\vc}{\mathrm{vc}}

\newcommand{\abs}[1]{\left| #1 \right|}

\newcommand{\myparagraph}[1]{\vspace{1em}\noindent\textbf{#1:}}

\newtheorem{theorem}{Theorem}
\newtheorem{lemma}{Lemma}

\begin{document}

\title{\Large Branch-and-Reduce Exponential/FPT Algorithms in Practice: \\ A Case Study of Vertex Cover}
\author{
Takuya Akiba\thanks{Department of Computer Science,
    Graduate School of Information Science and Technology,
    The University of Tokyo.
    \texttt{t.akiba@is.s.u-tokyo.ac.jp}}
\and
Yoichi Iwata\thanks{Department of Computer Science,
    Graduate School of Information Science and Technology,
    The University of Tokyo.
    \texttt{y.iwata@is.s.u-tokyo.ac.jp}}
}
\date{}
\maketitle

\begin{abstract}
We investigate the gap between theory and practice for exact branching algorithms.
In theory, branch-and-reduce algorithms
currently have the best time complexity for numerous important problems.
On the other hand, in practice,
state-of-the-art methods are based on different approaches,
and the empirical efficiency of such theoretical algorithms have seldom been investigated
probably because they are seemingly inefficient because of the plethora of complex reduction rules.
In this paper,
we design a branch-and-reduce algorithm for the vertex cover problem
using the techniques developed for theoretical algorithms
and compare its practical performance with other state-of-the-art empirical methods.
The results indicate that branch-and-reduce algorithms are actually quite practical
and competitive with other state-of-the-art approaches for several kinds of instances,
thus showing the practical impact of theoretical research on branching algorithms.

\end{abstract}

\section{Introduction}
Branching algorithms have been both theoretically and experimentally well studied to exactly solve NP-hard problems.
However, there is a gap between the theoretically fastest algorithms (i.e., those currently having the best time complexity) and
the empirically fastest algorithms (i.e., those currently with the best running time for popular benchmark instances).
In the theoretical research on exponential complexity or parameterized complexity of branching algorithms,
\emph{branch-and-reduce} methods, which involve a plethora of branching and reduction rules without using any lower bounds,
currently have the best time complexity for a number of important problems, such as \textsc{Independent Set} (or,
equivalently, \textsc{Vertex Cover})~\cite{DBLP:conf/isaac/XiaoN13,DBLP:journals/tcs/ChenKX10}, \textsc{Dominating Set}~\cite{DBLP:conf/iwpec/Iwata11}, and \textsc{Directed Feedback Vertex Set}~\cite{DBLP:conf/ictcs/Razgon07}.
On the other hand, in practice, \emph{branch-and-bound} methods that involve problem-specific lower bounds or LP-based
\emph{branch-and-cut} methods, which generate new cuts to improve the lower bounds, are often used.
While a number of complex rules have been developed to improve time complexity in theoretical research,
they seldom have been used for these empirical methods.

In this paper, we study the practical impact of theoretical research on branching algorithms.
As a benchmark problem, we choose \textsc{Vertex Cover} because it has been both theoretically and empirically well studied.
In this study, we design an algorithm that combines a variety of rules and lower bounds from several theoretical
studies.
We also develop new rules, called the \emph{packing branching} and \emph{packing reduction} rules,
which are inspired by these previous studies.
Then, we conduct experiments on a variety of instances and compare our algorithm with two state-of-the-art empirical
methods: a branch-and-cut method by a commercial integer programming solver, CPLEX, and a branch-and-bound method
called MCS~\cite{clique/mcs_walcom10}.
Although the rules in our algorithm are not designed for specific instances but are developed
for theoretical purposes, the results show that our algorithm is actually quite practical and competitive with other
state-of-the-art approaches for several cases.

\subsection{Relations to Theoretical Research on Exact Algorithms for Vertex Cover}
We introduce recent theoretical research on exact algorithms for \textsc{Vertex Cover}.
Two types of research exists: exact exponential algorithms, which analyze the exponential complexity
with regard to the number of vertices, and FPT algorithms, which introduce a parameter to the problem
and analyze the parameterized complexity with regard to both the parameter and the graph size.

First, we explain how the algorithms are designed and analyzed using a simple example.
Let us consider a very simple algorithm that selects a vertex $v$ and branches into two cases: either 1) including $v$ to the
vertex cover or 2) discarding $v$ while including its neighbors to the vertex cover.
Apparently, this algorithm runs in $O^*(2^n)$\footnote{The $O^*$ notation hides factors polynomial in $n$.} time.
Can we prove a better complexity?
The answer to this question would be No.
When a graph is a set of $n$ isolated vertices, the algorithm needs to branch on each vertex, which takes
$\Omega^*(2^n)$ time.
To avoid this worst case, we can add the following reduction rule: if a graph is not connected, we can solve each
connected component separately.
Now, we can assume that $v$ has a degree of at least one.
Then, after the second case of the branching, where $v$ is discarded and its neighbors are included, the number of
vertices to be considered decreases by at least two.
Thus, by solving the recurrence of $T(n)\leq T(n-1)+T(n-2)$, we can prove a complexity of $O^*(1.6181^n)$.
The worst case occurs when we continue to select a vertex of degree one.
Here, we note that if $n$ is at least three, a vertex of degree at least two always exists.
Thus, by adding the following branching rule, we can avoid this worst case: select a vertex of the maximum degree.
Now, we can assume that $v$ has a degree of at least two, and by solving the recurrence of $T(n)\leq T(n-1)+T(n-3)$, we can
prove the complexity of $O^*(1.4656^n)$.
We continue this process and create increasingly complex rules to avoid the worst case and improve the complexity.
Thus, currently, the theoretically fastest algorithms involve a number of complicated rules.
Although much of the current research uses a more sophisticated analytical tool called the measure and conquer
analysis~\cite{vc_reduction/measure_and_conquer},
the design process is basically the same.

As for exact exponential algorithms,
since Fomin, Grandoni, and Kratsch~\cite{vc_reduction/measure_and_conquer} gave an $O^*(1.2210^n)$-time algorithm by
developing the measure and conquer analysis, several improved algorithms have been
developed~\cite{DBLP:conf/fsttcs/KneisLR09,DBLP:journals/algorithmica/BourgeoisEPR12,DBLP:conf/isaac/XiaoN13}.
Since improving the complexity on sparse graphs is known to also improve the
complexity on general graphs~\cite{DBLP:journals/algorithmica/BourgeoisEPR12}, algorithms for sparse graphs also have
been well studied~\cite{DBLP:journals/jda/Razgon09,DBLP:journals/algorithmica/BourgeoisEPR12,DBLP:conf/isaac/XiaoN13}.
Among these algorithms, we use rules from the algorithm for general graphs by
Fomin~et~al.~\cite{vc_reduction/measure_and_conquer}, and the algorithm for sparse graphs by Xiao and
Nagamochi~\cite{DBLP:conf/isaac/XiaoN13}.
These rules are also contained in many of the other algorithms.
We also develop new rules inspired from the satellite rule presented by
Kneis, Langer, and Rossmanith~\cite{DBLP:conf/fsttcs/KneisLR09}.
Since our algorithm completely contains the rules of the algorithm by Fomin~et~al., our algorithm also can be proved to
run in $O^*(1.2210^n)$ time.

On FPT algorithms, \textsc{Vertex Cover} has been studied under various parameterizations.
Among them, the difference between the LP lower bound and the IP optimum is a recently developed parameter; however,
many interesting results have already been
obtained~\cite{DBLP:journals/toct/CyganPPW13,DBLP:journals/corr/abs-1203-0833,bip2/iwata14,DBLP:conf/soda/Wahlstrom14}.
While the exact exponential algorithms do not use any lower bounds to prune the search, with this parameterization,
the current fastest algorithms are based on the branch-and-bound method and use a (simple) LP lower bound to prune the
search.
Let $k$ be the parameter value, i.e., the difference between the LP lower bound and the IP optimum.
In our algorithm, we use an LP-based reduction rule from the $O(4^k m+m\sqrt{n})$\footnote{The algorithm runs in
$O(4^k m)$ time after computing the LP lower bound, which can be computed in $O(m\sqrt{n})$ time.}-time
algorithm by Iwata, Oka, and Yoshida~\cite{bip2/iwata14} and the LP lower bound.
Since we do not give the parameter to the algorithm, its search space may not be bounded by the
parameter.
However, if we were to use the iterative deepening strategy (we do not use it in this experiment), the running time of
our algorithm would also be bounded by $O^*(4^k)$.
The research on this parameterization also suggests that for many other problems, including \textsc{Odd Cycle
Transversal}, the fastest way to solve them is to reduce them into \textsc{Vertex Cover}.
Therefore, we conduct experiments on the graph reduced from an instance of \textsc{Odd Cycle Transversal}.
The results show that solving \textsc{Odd Cycle Transversal} through the reduction to \textsc{Vertex Cover}
strongly outperforms the state-of-the-art algorithm for \textsc{Odd Cycle Transversal}.
In our experiments, we used two different lower bounds that may give a better lower bound than LP relaxation.
We also investigated the parameterized complexity above these lower bounds.

\subsection{Organization}
In Section~\ref{sec:preliminaries}, we give definitions used in this paper.
The overview of our algorithm is described in Section~\ref{sec:overview}.
In Sections~\ref{sec:branching}, \ref{sec:reduction}, and \ref{sec:lower_bounds}, we give a list of the branching rules, the reduction
rules used in our algorithm, and the lower bounds we use, respectively.
We explain our experimental results in Section~\ref{sec:experiments}.
We investigate the parameterized complexity above the lower bounds we use in Appendix~\ref{sec:aboveLB}.
Finally, we conclude in Section~\ref{sec:conclusion}.

\section{Preliminaries}\label{sec:preliminaries}
Let $G=(V,E)$ be an undirected graph.
The \emph{degree} of a vertex $v$ is denoted by $d(v)$.
We denote the \emph{neighborhood} of a vertex $v$ by $N(v)=\{u\in V\mid uv\in E\}$
and the \emph{closed neighborhood} by $N[v]=N(v)\cup\{v\}$.
For a vertex subset $S\subseteq V$, we use the notation of $N(S)=\bigcup_{v\in S}N(v)\setminus S$ and
$N[S]=N(S)\cup S$.
The set of vertices at distance $d$ from a vertex $v$ is denoted by $N^d(v)$.
For a vertex subset $S\subseteq V$, we denote the induced subgraph on the vertex set $V\setminus S$ by $G-S$.
When $S$ is a single vertex $v$, we simply write $G-v$.

A \emph{vertex cover} of a graph $G$ is a vertex subset $C\subseteq V$ such that, for
any edge $e\in E$, at least one of its endpoints are contained in $C$.
We denote the set of all minimum vertex covers of $G$ by $\vc(G)$.
\textsc{Vertex Cover} is a problem to find a minimum vertex cover of a given graph.

\section{Algorithm Overview}\label{sec:overview}
The overview of our algorithm is described in Algorithm~\ref{alg:overview}.
For ease of presentation, the described algorithm only addresses the size of the minimum vertex cover.
However, obtaining the minimum vertex cover itself is not difficult.
Indeed, in our experiments, the minimum vertex cover is also computed,
and the time consumption to accomplish this is also accumulated.
The \emph{packing constraints} in the algorithm are created by our new branching and reduction rules.
They are not used to strengthen the LP relaxation, as in the branch-and-cut methods, but are used for the pruning and the
reduction.
We describe the details in Sections~\ref{sec:branching:packing} and \ref{sec:reduction:packing}.
We start the algorithm by setting $\mathcal{P}=\emptyset$, $c=0$, and $k=|V|$.
On each branching node, we first apply a list of reduction rules.
Then, we prune the search if the packing constraints are not satisfied or if the lower bound is
at least the size of the best solution we have.
If the graph is empty, we update the best solution.
If the graph is not connected, we separately solve each connected component.
Otherwise, we branch into two cases by applying the branching rule.
In our implementation, for time and space efficiency, we do not create new graphs after the branching but dynamically modify a
single graph.
As we note in the introduction, the algorithm can be proved to run in $O^*(1.2210^n)$.

{
\begin{algorithm}[tb]
\caption{The branch-and-reduce algorithm for \textsc{Vertex Cover}}
\label{alg:overview}
\begin{algorithmic}[1]
\INPUT a graph $G$, packing constraints $\mathcal{P}$, a current solution size $c$, and an upper
bound $k$
\Procedure{Solve}{$G,\mathcal{P},C,k$}
\State $(G,\mathcal{P},c)\leftarrow\Call{Reduce}{G,\mathcal{P},c}$
\If{$\Call{Unsatisfied}{\mathcal{P}}$}
	\Return $k$
\EndIf
\If{$c+\Call{LowerBound}{G}\geq k$}
	\Return $k$
\EndIf
\If{$G$ is empty}
	\Return $c$
\EndIf
\If{$G$ is not connected}
	\State $k'\leftarrow c$
	\ForAll{$(G_i,\mathcal{P}_i)\in\Call{Components}{G,\mathcal{P}}$}
		\State $k'\leftarrow k'+\Call{Solve}{G_i,\mathcal{P}_i,0,k-k'}$
	\EndFor
	\State \Return $k'$
\EndIf
\State $((G_1,\mathcal{P}_1,c_1),(G_2,\mathcal{P}_2,c_2))\leftarrow\Call{Branch}{G,\mathcal{P},c}$
\State $k\leftarrow\Call{Solve}{G_1,\mathcal{P}_1,c_1,k}$
\State $k\leftarrow\Call{Solve}{G_2,\mathcal{P}_2,c_2,k}$
\State \Return $k$
\EndProcedure
\end{algorithmic}
\end{algorithm}
}

\section{Branching Rules}\label{sec:branching}

\subsection{Vertex Selection}
In our main implementation, we completely use the same strategy as the one used in the theoretical exact exponential
algorithm by Fomin~et~al.~\cite{vc_reduction/measure_and_conquer} for selecting a vertex to branch on.
Basically, a vertex of the maximum degree is selected.
If there are multiple possibilities, we choose the vertex $v$ that minimizes the number of edges among $N(v)$.
In our experiments (Section~\ref{sec:experiments:lower_bounds}), we compare this strategy to the random selection strategy and the
minimum degree selection strategy.

\subsection{Mirror Branching}\label{sec:branching:mirror}
For a vertex $v$, a vertex $u\in N^2(v)$ is called a \emph{mirror} of $v$ if $N(v)\setminus N(u)$ induces a clique or
is an empty set.
We denote the set of mirrors for $v$ by $\mathcal{M}(v)$ and use the notation of $\mathcal{M}[v]=\mathcal{M}(v)\cup\{v\}$.
For the \emph{mirror branching} rule by Fomin~et~al.~\cite{vc_reduction/measure_and_conquer},
we branch into two cases: 1) including $\mathcal{M}[v]$ to the vertex cover or 2) discarding $v$ while including $N(v)$ to the
vertex cover.
In our implementation, we use this branching rule when the selected vertex $v$ has mirrors.

\subsection{Satellite Branching}
For a vertex $v$, a vertex $u\in N^2(v)$ is called a \emph{satellite} of $v$ if there exists a vertex $w\in N(v)$ such
that $N(w)\setminus N[v]=\{u\}$.
We denote the set of satellites for $v$ by $\mathcal{S}(v)$ and use the notation of
$\mathcal{S}[v]=\mathcal{S}(v)\cup\{v\}$.
Kneis~et~al.~\cite{DBLP:conf/fsttcs/KneisLR09} introduced the following \emph{satellite branching} rule for the case
in which there are no mirrors: 1) including $v$ to the vertex cover or 2) discarding $\mathcal{S}(v)$ while including
$N(\mathcal{S}(v))$ to the vertex cover.
In our implementation, instead of using this branching rule, we use a more powerful branching rule introduced in the
next subsection.

\subsection{Packing Branching}\label{sec:branching:packing}
Let $v$ be the selected vertex to branch on.
The proof outline for the correctness of the satellite branching is as follows.
If there exists a minimum vertex cover of $G$ that contains the vertex $v$, we can find it by searching for a minimum vertex
cover of $G-v$.
Otherwise, we can assume that no minimum vertex covers contain the vertex $v$.
If there exists a minimum vertex cover $C$ of $G$ that does not contain the vertex $v$ but contains a satellite
$u\in\mathcal{S}(v)$, then by discarding the vertex $w\in N(v)$ that satisfies $N(w)\setminus N[v]=\{u\}$ from $C$ and
including $v$ to $C$, we obtain a vertex cover that contains the vertex $v$ of the same size, which is a contradiction.

The key idea of satellite branching is that during the search for a minimum vertex cover that does not contain
the vertex $v$, we can assume that there are no minimum vertex covers that contain the vertex $v$.
To avoid the search of vertex covers from which we can confirm the existence of a vertex cover of the same size containing $v$,
we exploit this idea by explicitly creating constraints as follows.

For a vertex $w\in N(v)$, let $N^+(w)=N(w)\setminus N[v]$.
During the search for a minimum vertex cover that does not contain $v$, if we include all the vertices of
$N^+(w)$ to the vertex cover, by discarding the vertex $w$ and including the vertex $v$, we can obtain a vertex cover of
the same size.
Thus, in the search for a minimum vertex cover that does not contain $v$, for each vertex $w\in
N(v)$, we can introduce a constraint of $\sum_{u\in N^+(w)}x_u\leq |N^+(w)|-1$, where $x_u$ is a variable that indicates
whether the vertex $u$ is in the vertex cover (1) or not (0).
We call these constraints \emph{packing constraints}.
We keep and manage the constraints during the search;
when we include a vertex $v$ to the vertex cover, for each constraint that contains the variable $x_v$, we delete
the variable and decrease the right-hand side by one,
and when we delete a vertex $v$ from the graph without including it to the vertex cover,
for each constraint that contains the variable $x_v$, we delete the variable while keeping its right-hand side.
When some constraint is not satisfied at some node, i.e., the right-hand side of the constraint becomes negative, we
prune the subsequent search from the node.
We note that without packing constraints, we can prune the search only when the graph becomes empty or when the lower
bound exceeds the best solution we have found so far.

We can also introduce a packing constraint when we search for a minimum vertex cover that contains a vertex $v$.
If all the neighbors of $v$ are contained in the vertex cover, by discarding $v$, we can obtain a vertex cover of
smaller size.
Thus, in the search for a minimum vertex cover that contains $v$, we can introduce a constraint of $\sum_{u\in
N(v)}x_u\leq |N(v)|-1$.

Moreover, we can also use packing constraints for reductions.
If the right-hand side of a constraint becomes zero but the left-hand side contains a variable $x_u$, we can delete the
vertex $u$ from the graph while including its neighbors $N(u)$ to the vertex cover.
The satellite branching corresponds to the case that $N^+(w)$ is a single vertex set.
In Section~\ref{sec:reduction:packing}, we introduce more sophisticated reduction rules to exploit packing constraints.

We note that the total size of packing constraints scales at most linearly with the graph size because we create at most one
constraint for each vertex $w$ and the size of each constraint is at most the degree of the corresponding vertex.
Thus explicitly keeping all the constraints does not seriously affect the computation time.
We also note that packing constraints are auxiliary;
i.e., our objective is not to search for a minimum vertex cover under the constraints
but to search for a minimum vertex cover or conclude that there exists a minimum
vertex cover not satisfying the constraints (which can be found in another case of the branching).

\section{Reduction Rules}
\label{sec:reduction}

\subsection{Reductions from Exponential Algorithms}\label{sec:reduction:exponential}

First, we introduce four reduction rules from the exact exponential algorithm by
Fomin~et~al.~\cite{vc_reduction/measure_and_conquer}.
Three of them are quite simple.
The first one is the \emph{components} rule.
When a graph is not connected, we can solve for each component separately.
The second one is the \emph{degree-1} rule.
If a graph contains a vertex of degree at most one, there always exists a
minimum vertex cover that does not contain the vertex.
Therefore, we can delete it and include its neighbors to the vertex cover.
The third one is the \emph{dominance} rule.
We say a vertex $v$ \emph{dominates} a vertex $u$ if $N[u]\subseteq N[v]$.
If a vertex $v$ dominates some vertex, there always exists a minimum vertex cover that contains $v$.
Therefore, we can include it to the vertex cover.
We note that the degree-1 rule is completely
contained in the components rule and the
dominance rule.
However, it is still useful because its computational cost is smaller in practice.
The final rule, \textit{degree-2 folding}, is somewhat tricky.
It removes a vertex of degree two and its neighbors while introducing a new
vertex as in the following lemma.
\begin{lemma}[Degree-2
Folding~\cite{vc_reduction/measure_and_conquer}]
Let $v$ be a vertex of degree two whose two neighbors are not adjacent, and let $G'$
be a graph obtained from $G$ by removing $N[v]$, introducing a new vertex $w$ which is connected to $N^2(v)$.
Then, for any $C'\in\vc(G')$, the following $C$ is in $\vc(G)$:
\begin{equation*}
C=\begin{cases}
C'\cup\{v\}&(w\not\in C'),\\
(C'\setminus\{w\})\cup N(v)&(w\in C').
\end{cases}
\end{equation*}
\end{lemma}

\subsection{Reductions from FPT Algorithms}\label{sec:reduction:lp}
We use the \emph{LP-based reduction} rule which was first developed by Nemhauser
and Trotter~\cite{NemhauserT75}, and then, strengthened by Iwata, Oka, and
Yoshida~\cite{bip2/iwata14}.

The LP relaxation of \textsc{Vertex Cover} can be written as follows:
\begin{align*}
  &\text{minimize}\displaystyle \sum_{v\in V}x_v\\
  &\begin{array}{lll}
    \text{s.t.}&x_u+x_v\geq 1 & \text{for }uv\in E,\\
    & x_v \geq 0 & \text{for } v \in V. \\
  \end{array}
\end{align*}
Nemhauser and Trotter~\cite{NemhauserT75} showed that the above LP has the following two properties:
\begin{itemize}
  \item There exists an optimal solution such that each variable takes a value $0$, $1$, or
  $\frac{1}{2}$ (\emph{Half-integrality}).
  \item If a variable $x_v$ takes an integer value in an optimal LP solution, there always exists an optimal integer
  solution in which $x_v$ takes the same value (\emph{Persistency}).
\end{itemize}
They also showed that a half-integral optimal solution of the LP relaxation can be computed by reducing it to the bipartite
matching problem as follows.
From the input graph $G=(V,E)$, we construct a bipartite graph $\bar{G}=(L_V\cup R_V, \bar{E})$ such that:
\begin{align*}
L_V&=\{l_v\mid v\in V\},\\
R_V&=\{r_v\mid v\in V\},\\
\bar{E}&=\{l_ur_v\mid uv\in E\}\cup\{l_vr_u\mid uv\in E\}.
\end{align*}
Let $\bar{C}$ be a minimum vertex cover of the bipartite graph $\bar{G}$,
which can be computed in linear-time from
a maximum matching of $\bar{G}$.
Then, the value of $x^*_v$ of the half-integral optimal solution is determined as follows:
\begin{align*}
x^*_v=\begin{cases}
0&(l_v,r_v\not\in\bar{C}),\\
1&(l_v,r_v\in\bar{C}),\\
\frac{1}{2}&(\text{otherwise}).
\end{cases}
\end{align*}

From the persistency of the LP relaxation, we can fix the integral part of a half-integral optimal solution;
i.e., we can discard the vertices of value zero and include the vertices of value one to the vertex cover.

Iwata, Oka, and Yoshida~\cite{bip2/iwata14} presented an algorithm for computing a half-integral
optimal solution of the LP relaxation whose half-integral part is minimal.
We call such an optimal solution an \emph{extreme} optimal solution.
The algorithm runs in linear time after computing a maximum matching of the graph described above.
Instead of using an arbitrary half-integral optimal solution, we use a half-integral extreme optimal solution
computed by this algorithm.
If $x^*$ is a half-integral extreme optimal solution, the graph induced on the half-integral part has a unique optimal
solution of the all-half vector.
Thus no more vertices can be deleted by using other optimal solutions.
We note that the famous \emph{crown reduction} rule [Abu-Khzam~et~al., 2004] is
completely contained in the LP reduction rule if we use an extreme optimal solution.
Thus, we do not use it in our algorithm.

In our implementation, we compute the maximum matching by using the Hopcroft-Karp algorithm, which runs in
$O(|E|\sqrt{|V|})$ time.
When the graph is changed by reductions or branchings, we do not recompute the maximum matching from scratch but
modify the current non-maximum matching to the maximum one by searching for augmenting paths of the residual graph.

\subsection{Reductions from Exponential Algorithms for Sparse Graphs}\label{sec:reduction:sparse}
Now, we introduce the four reduction rules that appeared in the exact exponential algorithm for sparse graphs by
Xiao and Nagamochi~\cite{vc_reduction/nagamochi2013}.
These rules are very complicated, but as we see in Section~\ref{sec:experiments}, they are quite
useful in practice.

The first rule, \emph{unconfined}, is a generalization of the dominance
and the satellite rules by Kneis~et~al.~\cite{DBLP:conf/fsttcs/KneisLR09}.
A vertex $v$ is called unconfined if the following procedure returns yes:
\begin{enumerate}
  \item Let $S=\{v\}$.
  \item Find $u\in N(S)$ such that $|N(u)\cap S|=1$ and $|N(u)\setminus N[S]|$ is minimized.
  \item If there is no such vertex, return no.
  \item If $N(u)\setminus N[S]=\emptyset$, return yes.
  \item If $N(u)\setminus N[S]$ is a single vertex $w$, go back to line 2 by
  adding $w$ to $S$.
  \item Return no.
\end{enumerate}
For any unconfined vertex $v$, there always exists a minimum vertex cover that contains $v$.
Thus, we can include it to the vertex cover.

The second rule, \emph{twin}, is similar to the degree-2 folding rule.
Two vertices $u$ and $v$ are called a twin if $N(u)=N(v)$ and $d(u)=d(v)=3$.
If there is a twin, we can make the graph smaller, as in the following lemma.
\begin{lemma}[Twin~\cite{vc_reduction/nagamochi2013}]
Let $u$ and $v$ be a twin.
If there exists an edge among $N(u)$, for any $C'\in\vc(G-N[\{u,v\}])$, $C'\cup N(u)\in\vc(G)$.
Otherwise, let $G'$ be a graph obtained from $G$ by removing $N[\{u,v\}]$, introducing a new vertex $w$ connected to $N^2(u)\setminus\{v\}$.
Then, for any $C'\in\vc(G')$, the following $C$ is in $\vc(G)$:
\begin{equation*}
C=\begin{cases}
C'\cup\{u,v\}&(w\not\in C'),\\
(C'\setminus\{w\})\cup N(u)&(w\in C').
\end{cases}
\end{equation*}
\end{lemma}

Now, we introduce the notion of \textit{alternative}.
Two subsets of vertices $A$ and $B$ are called alternatives if $|A|=|B|\geq 1$ and there exists a minimum vertex cover
$C$ that satisfies $C\cap (A\cup B)=A\text{ or }B$.
The third and fourth rules are special cases of the alternative.
Let $u, v$ be adjacent vertices such that $N(v)\setminus\{u\}$ induces a complete graph.
Then, $\{u\}$ and $\{v\}$ are alternative sets (called a \emph{funnel}).
Let $a_1b_1a_2b_2$ be a chordless 4-cycle such that the degree of each vertex is at least three.
Let $A=\{a_1,a_2\}$ and $B=\{b_1,b_2\}$.
If it holds that $N(A)\cap N(B)=\emptyset$, $|N(A)\setminus B|\leq 2$, and $|N(B)\setminus A|\leq 2$,
then $A$ and $B$ are alternatives (called a \emph{desk}).
If there is a funnel or a desk, we can remove it by the following lemma.
\begin{lemma}[Alternative~\cite{vc_reduction/nagamochi2013}]
Let $A, B$ be alternative subsets of vertices, and $G'$ be a graph obtained from $G$ by removing
$(N(A)\cap N(B))\cup A\cup B$ and introducing an edge between every two
nonadjacent vertices $u\in N(A)\setminus N[B]$ and $v\in N(B)\setminus N[A]$.
Then, for any $C'\in\vc(G')$, the following $C$ is in $\vc(G)$:
\begin{equation*}
C=\begin{cases}
C'\cup(N(A)\cap N(B))\cup A&(N(B)\setminus N[A]\subseteq C'),\\
C'\cup(N(A)\cap N(B))\cup B&(N(A)\setminus N[B]\subseteq C').
\end{cases}
\end{equation*}
\end{lemma}

\subsection{Packing Reductions}\label{sec:reduction:packing}
In Section~\ref{sec:branching:packing}, we introduce the branching rule that creates auxiliary constraints, called
packing constraints, and introduce the simple reduction rule on the basis of these constraints.
In this section, we introduce more sophisticated reduction rules to exploit packing constraints.
Let $\sum_{v\in S}x_v\leq k$ be a packing constraint such that $S$ is nonempty.

The first rule is for the case in which $k$ is zero.
To satisfy the constraint, we cannot include any vertices in $S$ to the vertex cover.
Thus, if there is an edge among $S$, we can prune the subsequent search.
Otherwise, we can delete $S$ from the graph while including $N(S)$ to the vertex cover.
Here, we can introduce additional packing constraints.
Let $u$ be a vertex such that $N(u)\cap S$ is a single vertex $w$, and let $N^+(u)=N(u)\setminus N[S]$.
If we include all the vertices of $N^+(u)$ to the vertex cover, by discarding the vertex $u$ and including
$w$, we can obtain a vertex cover of the same size that does not satisfy the constraint of $\sum_{v\in S}x_v\leq 0$.
Thus, we can introduce a new constraint of $\sum_{v\in N^+(u)}x_v\leq |N^+(u)|-1$.

The second rule is for the case in which $k$ is positive.
Let $u\not\in S$ be a vertex such that $|S\cap N(u)|>k$.
If we do not include $u$ to the vertex cover, all the vertices of $N(u)$ must be contained in the vertex cover.
Thus, the constraint is not satisfied.
Therefore, we can include $u$ to the vertex cover.
Moreover, if at least $|N(u)|-1$ vertices of $N(u)$ are included to the vertex cover, by discarding $u$ and including
the remaining vertex of $N(u)$, we can obtain a vertex cover of the same size that does not satisfy the constraint.
Thus, we can introduce a new constraint of $\sum_{v\in N(u)}x_v\leq |N(u)|-2$.

When we also use reduction rules such as the degree-2 folding, which modifies the graph by deleting some vertices
and creating new vertices, the deleted vertices might be included to the vertex cover later on.
In that case, we revert the modification until all the vertices in the constraint are recovered and then check the
constraint.

\section{Lower Bounds}\label{sec:lower_bounds}

We introduce several lower bounds that can be easily computed.
In our main implementation, we take the maximum of them as a lower bound.

\subsection{Clique Cover}\label{sec:lower_bounds:clique}
A set of disjoint cliques $C_1,\ldots,C_k$ is called \emph{clique cover} if it covers all the vertices.
For a clique cover $C_1,\ldots,C_k$, the value $\sum_{i=1}^k(|C_i|-1)=|V|-k$ gives a lower bound for
the size of the minimum vertex cover.

In our implementation, we compute a clique cover greedily as follows.
First, we sort the vertices by ascending order of their degrees and initiate a set of cliques $\mathcal{C}$ to be an
empty set.
Then, for each vertex $v$, we search for a clique $C\in\mathcal{C}$ to which $v$ can be added.
If there are multiple possible cliques, we choose the one with maximum size.
If there are no such cliques, we add a clique of the single vertex $v$ to $\mathcal{C}$.
Since it takes only $O(d(v))$ time for each vertex $v$, the algorithm runs in linear time in total.

This lower bound is also used in the state-of-the-art branch-and-bound algorithm \emph{MCS}~\cite{clique/mcs_walcom10}.
MCS computes a clique cover using a more sophisticated strategy to obtain a better lower bound.
However, it does not scale for large graphs.

\subsection{LP Relaxation}\label{sec:lower_bounds:lp}
The optimal value of the LP relaxation gives a lower bound for the size of the minimum vertex cover.
After the LP-based reduction, the remaining graph admits a half-integral optimal solution of value $\frac{|V|}{2}$.
This lower bound has been used in FPT algorithms parameterized by the difference between LP lower bounds and the IP
optimum~\cite{DBLP:journals/corr/abs-1203-0833,bip2/iwata14}.

\subsection{Cycle Cover}\label{sec:lower_bounds:cycle}
A set of disjoint cycles $C_1,\ldots,C_k$ is called \emph{cycle cover} if it covers all the vertices.
Here, two adjacent vertices are considered as a cycle of length two, but a single vertex does not form a cycle of
length one.
For a cycle cover $C_1,\ldots,C_k$, the value $\sum_{i=1}^k\left\lceil\frac{|C_i|}{2}\right\rceil$ gives a
lower bound for the size of the minimum vertex cover.

We do not have to compute a cycle cover from scratch.
After the LP-based reduction, the bipartite graph $\bar{G}$ of Section~\ref{sec:reduction:lp} admits a perfect
matching.
Thus, by taking an edge $uv$ for each edge $l_u r_v$ in the perfect matching, we can obtain a cycle cover of the graph
$G$ in $O(|V|)$ time.
Since the optimal value of the LP relaxation is $\frac{|V|}{2}=\sum_{i=1}^{k}\frac{|C_i|}{2}$, the lower bound given by
this cycle cover is never worse than the LP optimum.
Let $v_1,...,v_n$ be vertices forming a cycle.
If there are four vertices $\{v_i, v_{i+1}, v_j, v_{j+1}\}$ with edges $v_iv_{j+1}$ and $v_jv_{i+1}$, we can split the cycle into two smaller cycles.
In our implementation, if it is possible to split a cycle of even length into two smaller cycles
of odd length, we split it to improve the lower bound.

\section{Experiments}\label{sec:experiments}
\subsection{Setup}
Experiments were conducted on a machine with Intel
Xeon X5670 (2.93 GHz) and 48GB of main memory
running Linux 2.6.18.
C++ programs were compiled using gcc 4.8.2 with \texttt{-O3} option.
Java programs were executed with JRE 1.8.0.
All the timing results were sequential.
We set the time limit for each execution as 24 hours.

\subsubsection{Instances}
As for problem instances, we focused on real large sparse networks.
Computing small vertex covers on these networks
is important for graph indexing methods~\cite{app/vc_index,app/k-path}.
We also used instances from DIMACS Implementation Challenge and the \OCT problem.
Directions of edges are ignored and self-loops were removed beforehand.
The detailed description of the three sets of graphs are as follows.

\myparagraph{Real Sparse Networks} We focused on real large sparse networks
such as social networks, web graphs, computer networks and road networks.
They were obtained from the Stanford Large Network Dataset Collection\footnote{\url{http://snap.stanford.edu/data/}}, Koblenz Network Collection\footnote{\url{http://konect.uni-koblenz.de/}}, and Laboratory for Web Algorithmics\footnote{\url{http://law.di.unimi.it/datasets.php}}~\cite{datasets/webgraph1,datasets/webgraph2}.

\myparagraph{DIMACS Instances}
DIMACS Instances are those from DIMACS Implementation Challenge on the maximum clique problem~\cite{datasets/dimacs}.
They consist of artificial synthetic graphs and problems reduced from other problems.
We used complement graphs of them as \textsc{Vertex Cover} instances.
Since they are originally dense graphs and have at most thousands of vertices,
explicitly considering complement graphs is feasible
and has been often done for benchmarking algorithms for
\textsc{Vertex Cover} and \textsc{Minimum Independent Set}.
Indeed, these complement graphs are also available online for these problems,
and we downloaded them\footnote{\url{http://www.cs.hbg.psu.edu/txn131/vertex_cover.html}}.

\myparagraph{Instances from \OCT}
An \emph{odd cycle transversal} of a graph $G$ is a vertex subset $S\subseteq V$ such that $G-S$ becomes a bipartite
graph.
\textsc{Odd Cycle Transversal} is a problem in which a minimum odd cycle transversal of a given graph has to be found.
The problem is known to be reduced to \textsc{Vertex Cover} in linear time as
follows~\cite{DBLP:journals/corr/abs-1203-0833}.
From the input graph $G=(V,E)$, we construct a graph $\hat{G}=(L_V\cup R_V,\hat{E})$ where
$L_V=\{l_v\mid v\in V\}$, $R_V=\{r_v\mid v\in V\}$, and
$\hat{E}=\{l_ul_v\mid uv\in E\}\cup\{r_ur_v\mid uv\in E\}\cup\{l_v r_v\mid v\in V\}$
Let $\hat{C}$ be a minimum vertex cover of $\hat{G}$.
Then, a minimum odd cycle transversal can be computed by taking vertices $v\in V$ such that both $l_v$ and
$r_v$ are in $\hat{C}$.

We used real \textsc{Odd Cycle Transversal} instances from bioinformatics,
which formulates the \textsc{Minimum Site Removal} problem\footnote{\url{http://www.user.tu-berlin.de/hueffner/occ/}}~\cite{oct/huffner09}.

\subsubsection{Methods}
We generally compare the three algorithms for \textsc{Vertex Cover} based on different approaches:
B\&R, CPLEX and MCS~\cite{clique/mcs_walcom10}.
For instances from \OCT, we also include the results of the algorithm for directly solving \OCT
by H{\"u}ffner~\cite{oct/huffner09}.

\myparagraph{B\&R} \emph{B\&R} is the branch-and-reduce algorithm stated above, which is implemented in Java.
Unless mentioned otherwise,
all the branching rules (Section~\ref{sec:branching}),
all the reduction rules (Section~\ref{sec:reduction}),
and all the lower bounds (Section~\ref{sec:lower_bounds})
were used.

\myparagraph{CPLEX} \emph{IBM ILOG CPLEX Optimization Studio (CPLEX)} is a state-of-the-art commercial optimization software package.
We used version 12.6 and formulated \textsc{Vertex Cover} through integer programming.
To exactly obtain the minimum vertex cover,
we set \texttt{mip tolerances mipgap} and \texttt{mip tolerances absmipgap} as zero
and switched on \texttt{emphasis numerical}.
Nevertheless, CPLEX did not produce truly optimal solutions for some instances,
probably because of numerical precision issues\footnote{We also tested modification of the feasibility tolerance parameter in its simplex routine
(\texttt{simplex tolerances feasibility}).
Indeed, the results were quite sensitive to this parameter,
which implies numerical precision issues.
We observed that the best results were produced
by the default value $10^{-6}$ in almost all the cases,
and thus, the default value was used for this parameter.
Consequently, we switched on the numerical precision emphasis parameter (\texttt{emphasis numerical}).
While it improved the results to some extent,
in some instances the results were larger than ours,
though we confirmed that our smaller solutions were, indeed, vertex covers.}.
These results are presented in our tables in parentheses.

\myparagraph{MCS} \emph{MCS}~\cite{clique/mcs_walcom10} is a state-of-the-art branch-and-bound algorithm
for the \textsc{Maximum Clique} problem.
We used this algorithm for computing minimum vertex cover by virtually considering complement graphs.
The algorithm is tailored to DIMACS instances
and uses the \emph{greedy coloring technique} to obtain good lower bounds.
The algorithm never applies any reductions.
It was implemented by the authors in C++.

\myparagraph{H{\"u}ffner}
This is the state-of-the-art algorithm by H{\"u}ffner
for directly solving \OCT~\cite{oct/huffner09}.
This algorithm is based on an FPT algorithm by Reed, Smith and Vetta~\cite{oct/reed2004}
using the \emph{iterative compression technique}.

\subsection{Algorithm Comparison}
The experimental results on real sparse networks and DIMACS instances
are shown in Table~\ref{tbl:comparison_vc}.
For each instance,
the table lists the number of vertices ($\abs{V}$),
the number of edges ($\abs{E}$),
and results of the three methods.
For each method,
besides time consumption in seconds (T),
the number of branches (\#B) are described.
For CPLEX, the number of introduced cuts (\#C) is also denoted.

{
\tabcolsep = 1.5mm
\begin{table}[p]
\center
\caption{Performance comparison of algorithms for \textsc{Vertex Cover};
T denotes running time (in seconds), \#B denotes the number of branches, and \#C denotes the number of introduced cuts.}
\label{tbl:comparison_vc}
\fontsize{8pt}{0pt}\selectfont
\begin{tabular}{lrr|rr|rrr|rr}
\toprule
\multicolumn{3}{c|}{\textbf{Instance}} & \multicolumn{2}{c|}{\textbf{B\&R}}
 & \multicolumn{3}{c|}{\textbf{CPLEX}} &
\multicolumn{2}{c}{\textbf{MCS~\cite{clique/mcs_walcom10}}} \\
\multicolumn{1}{c}{\textbf{Name}} &
\multicolumn{1}{c}{\textbf{$\abs{V}$}} &
\multicolumn{1}{c|}{\textbf{$\abs{E}$}} &
\multicolumn{1}{c}{\textbf{T}} &
\multicolumn{1}{c|}{\textbf{\#B}} &
\multicolumn{1}{c}{\textbf{T}} &
\multicolumn{1}{c}{\textbf{\#B}} &
\multicolumn{1}{c|}{\textbf{\#C}} &
\multicolumn{1}{c}{\textbf{T}} &
\multicolumn{1}{c}{\textbf{\#B}} \\
\midrule
\multicolumn{4}{l}{\textbf{\textit{Social Networks:}}} \\
ca-GrQc & 5,242 & 14,484 & 0.1 & 0 & 0.1 & 0 & 0 & 61.5 & 4,351 \\
ca-HepTh & 9,877 & 25,973 & 0.2 & 0 & 0.2 & 0 & 0 & 2,181.3 & 33,183 \\
ca-CondMat & 23,133 & 93,439 & 0.1 & 0 & 0.8 & 0 & 0 & 36,540.5 & 107,891 \\
wiki-Vote & 7,115 & 100,762 & 0.0 & 0 & 0.3 & 0 & 0 & -- & -- \\
ca-HepPh & 12,008 & 118,489 & 0.1 & 0 & 0.7 & 0 & 0 & 1,803.7 & 26,231 \\
email-Enron & 36,692 & 183,831 & 0.6 & 0 & 1.4 & 0 & 0 & -- & -- \\
ca-AstroPh & 18,772 & 198,050 & 0.1 & 0 & 1.6 & 0 & 0 & -- & -- \\
email-EuAll & 265,214 & 364,481 & 0.1 & 0 & 2.0 & 0 & 0 & -- & -- \\
soc-Epinions1 & 75,879 & 405,740 & 1.1 & 0 & 1.5 & 0 & 0 & -- & -- \\
soc-Slashdot0811 & 77,360 & 469,180 & 0.2 & 0 & 1.9 & 0 & 0 & -- & -- \\
soc-Slashdot0902 & 82,168 & 504,230 & 0.2 & 0 & 2.6 & 0 & 0 & -- & -- \\
dblp-2010 & 300,647 & 807,700 & 2.0 & 0 & 13.0 & 0 & 0 & -- & -- \\
youtube-links & 1,138,499 & 2,990,443 & 1.0 & 0 & 14.8 & 0 & 0 & -- & -- \\
dblp-2011 & 933,258 & 3,353,618 & 4.8 & 0 & 89.7 & 0 & 0 & -- & -- \\
wiki-Talk & 2,394,385 & 4,659,565 & 1.3 & 0 & 64.6 & 0 & 0 & -- & -- \\
petster-cat & 149,700 & 5,448,197 & 3.0 & 0 & 1,713.4 & 0 & 3,042 & -- & -- \\
petster-dog & 426,820 & 8,543,549 & 5.8 & 3 & 1,702.6 & 0 & 9,522 & -- & -- \\
youtube-u-growth & 3,223,589 & 9,376,594 & 2.6 & 0 & 155.1 & 0 & 0 & -- & -- \\
flickr-links & 1,715,255 & 15,555,041 & 2.1 & 0 & 121.6 & 0 & 132 & -- & -- \\
petster-carnivore & 623,766 & 15,695,166 & 3.6 & 0 & 6,240.6 & 0 & 67,378 & -- & -- \\
libimseti & 220,970 & 17,233,144 & 1,642.8 & 472 & -- & -- & -- & -- & -- \\
soc-pokec & 1,632,803 & 22,301,964 & -- & -- & -- & -- & -- & -- & -- \\
flickr-growth & 2,302,925 & 22,838,276 & 2.8 & 0 & 239.6 & 0 & 323 & -- & -- \\
soc-LiveJournal1 & 4,847,571 & 42,851,237 & 11.5 & 36 & 23,876.5 & 0 & 114,581 & -- & -- \\
hollywood-2009 & 1,107,243 & 56,375,711 & 26.8 & 0 & -- & -- & -- & -- & -- \\
hollywood-2011 & 1,985,306 & 114,492,816 & 50.5 & 0 & -- & -- & -- & -- & -- \\
orkut-links & 3,072,441 & 117,185,083 & -- & -- & -- & -- & -- & -- & -- \\
\midrule \multicolumn{4}{l}{\textbf{\textit{Web Graphs:}}} \\
web-NotreDame & 325,729 & 1,090,108 & 13.4 & 4,266 & (181.2) & 478 & 539 & -- & -- \\
web-Stanford & 281,903 & 1,992,636 & 67,270.3 & 55,865,269 & (1,836.1) & 1,006 & 36,237 & -- & -- \\
baidu-related & 415,641 & 2,374,044 & 2.1 & 8 & 979.4 & 0 & 31,744 & -- & -- \\
cnr-2000 & 325,557 & 2,738,969 & -- & -- & 4,124.4 & 503 & 18,579 & -- & -- \\
web-Google & 875,713 & 4,322,051 & 1.4 & 10 & 113.3 & 0 & 332 & -- & -- \\
web-BerkStan & 685,230 & 6,649,470 & 142.3 & 42,270 & 6,777.4 & 970 & 66,277 & -- & -- \\
in-2004 & 1,382,869 & 13,591,473 & 3.5 & 668 & 9,445.6 & 484 & 58,954 & -- & -- \\
hudong-internal & 1,984,484 & 14,428,382 & 1.9 & 5 & 141.7 & 0 & 21 & -- & -- \\
eu-2005 & 862,664 & 16,138,468 & -- & -- & 39,847.5 & 484 & 56,496 & -- & -- \\
baidu-internal & 2,141,300 & 17,014,946 & 2.4 & 0 & 161.7 & 0 & 63 & -- & -- \\
indochina-2004 & 7,414,768 & 150,984,819 & -- & -- & -- & -- & -- & -- & -- \\
uk-2002 & 18,484,117 & 261,787,258 & -- & -- & -- & -- & -- & -- & -- \\
\midrule \multicolumn{4}{l}{\textbf{\textit{Computer Networks:}}} \\
p2p-Gnutella08 & 6,301 & 20,777 & 0.2 & 0 & 0.1 & 0 & 0 & 412.7 & 7,882 \\
p2p-Gnutella09 & 8,114 & 26,013 & 0.1 & 0 & 0.1 & 0 & 0 & 831.2 & 5,540 \\
p2p-Gnutella06 & 8,717 & 31,525 & 0.1 & 0 & 0.1 & 0 & 0 & -- & -- \\
p2p-Gnutella05 & 8,846 & 31,839 & 0.0 & 0 & 0.1 & 0 & 0 & -- & -- \\
p2p-Gnutella04 & 10,876 & 39,994 & 0.0 & 0 & 0.1 & 0 & 0 & -- & -- \\
p2p-Gnutella25 & 22,687 & 54,705 & 0.2 & 0 & 0.2 & 0 & 0 & -- & -- \\
p2p-Gnutella24 & 26,518 & 65,369 & 0.0 & 0 & 0.2 & 0 & 0 & -- & -- \\
p2p-Gnutella30 & 36,682 & 88,328 & 0.1 & 0 & 0.3 & 0 & 0 & -- & -- \\
p2p-Gnutella31 & 62,586 & 147,892 & 0.1 & 0 & 0.6 & 0 & 0 & -- & -- \\
as-Skitter & 1,696,415 & 11,095,298 & 2,769.8 & 2,123,545 & (1,343.0) & 522 & 1,812 & -- & -- \\
\midrule \multicolumn{4}{l}{\textbf{\textit{Road Networks:}}} \\
roadNet-PA & 1,088,092 & 1,541,898 & -- & -- & 1,699.0 & 1,028 & 51,819 & -- & -- \\
roadNet-TX & 1,379,917 & 1,921,660 & -- & -- & 2,191.6 & 990 & 58,368 & -- & -- \\
roadNet-CA & 1,965,206 & 2,766,607 & -- & -- & 3,043.9 & 495 & 91,847 & -- & -- \\
\midrule \multicolumn{4}{l}{\textbf{\textit{DIMACS Instances:}}} \\
c-fat200-1 & 200 & 18,366 & 1.0 & 1 & 11.3 & 0 & 141 & 0.0 & 3 \\
c-fat200-5 & 200 & 11,427 & 0.2 & 1 & 6.2 & 0 & 48 & 0.0 & 27 \\
hamming10-2 & 1,024 & 5,120 & 0.2 & 0 & 0.1 & 0 & 0 & 0.1 & 512 \\
keller4 & 171 & 5,100 & 4.2 & 4,201 & 24.5 & 5,822 & 60 & 0.0 & 6,800 \\
MANN\_a27 & 378 & 702 & 2.3 & 1,396 & 4.4 & 2,981 & 25 & 0.6 & 9,164 \\
hamming8-4 & 256 & 11,776 & 16.9 & 14,690 & 2.3 & 0 & 119 & 0.1 & 19,567 \\
MANN\_a45 & 1,035 & 1,980 & 77.7 & 123,907 & 15.9 & 6,179 & 41 & 118.1 & 249,186 \\
p\_hat700-2 & 700 & 122,922 & 701.2 & 343,613 & -- & -- & -- & 5.6 & 275,676 \\
p\_hat1500-1 & 1,500 & 839,327 & 9,959.7 & 2,890,004 & -- & -- & -- & 4.3 & 834,181 \\
sanr200\_0.9 & 200 & 2,037 & 702.7 & 1,690,472 & 1,523.7 & 603,714 & 37 & 39.2 & 3,369,435 \\
sanr400\_0.7 & 400 & 23,931 & 24,365.0 & 41,203,755 & -- & -- & -- & 188.7 & 30,154,732 \\
p\_hat700-3 & 700 & 61,640 & -- & -- & -- & -- & -- & 2,408.7 & 88,791,027 \\
C250.9 & 250 & 3,141 & 53,125.4 & 98,016,830 & -- & -- & -- & 2,740.3 & 223,414,645 \\
gen400\_p0.9\_55 & 400 & 7,980 & -- & -- & 7,049.4 & 782,289 & 17 & 20,818.2 & 981,661,757 \\
brock800\_1 & 800 & 112,095 & -- & -- & -- & -- & -- & 10,025.3 & 1,273,480,056 \\
san400\_0.9\_1 & 400 & 7,980 & -- & -- & 5.0 & 0 & 26 & 69,139.5 & 3,762,277,624 \\
C2000.5 & 2,000 & 999,164 & -- & -- & -- & -- & -- & 70,723.5 & 11,749,950,425 \\
C500.9 & 500 & 12,418 & -- & -- & -- & -- & -- & -- & -- \\
keller6 & 3,361 & 1,026,582 & -- & -- & -- & -- & -- & -- & -- \\
gen400\_p0.9\_65 & 400 & 7,980 & -- & -- & 36.7 & 583 & 15 & -- & -- \\
\bottomrule
\end{tabular}
\end{table}
}

We first observe that B\&R and CPLEX clearly outperform MCS on real sparse networks.
Also, except for road networks, B\&R is generally comparable with CPLEX.
B\&R solves several cases that CPLEX fails to solve within the time limit,
such as libimseti, hollywood-2009, and hollywood-2011.
Moreover, for some of the other instances,
such as petster-cat, soc-LiveJournal1, web-Google, and in-2004,
B\&R is orders of magnitude faster than CPLEX.
On the other hand, for a few web graph such as cnr-2000 and eu-2005,
only CPLEX gave an answer within the time limit.

In contrast, on DIMACS instances, as it is tailored to these instances,
MCS generally works better.
The performances of B\&R and CPLEX are comparable.
For example, B\&R solved some of the \emph{p\_hat} instances and \emph{sanr} instances
that CPLEX could not solve,
but CPLEX solved some of the \emph{gen} instances that B\&R could not solve.

Table~\ref{tbl:comparison_oct} lists the results on instances from the \textsc{Odd Cycle Transversal} problem.
For H{\"u}ffner's algorithm,
we describe the number of augmentations (\#A) instead of the number of branches.
We observe that B\&R, CPLEX and MCS strongly outperform H{\"u}ffner's algorithm.

{
\begin{table}[p!]
\center
\caption{Performance comparison of algorithms for \OCT.}
\label{tbl:comparison_oct}
\fontsize{8pt}{0pt}\selectfont
\begin{tabular*}{\columnwidth}{@{\extracolsep{\fill}}lrr|rr|rrr|rr|rr}
\toprule
\multicolumn{3}{c|}{\textbf{Instance}} & \multicolumn{2}{c|}{\textbf{B\&R}}
 & \multicolumn{3}{c|}{\textbf{CPLEX}} & \multicolumn{2}{c|}{\textbf{MCS~\cite{clique/mcs_walcom10}}}
 & \multicolumn{2}{c}{\textbf{H{\"u}ffner~\cite{oct/huffner09}}} \\
\multicolumn{1}{c}{\textbf{Name}} &
\multicolumn{1}{c}{\textbf{$\abs{V}$}} &
\multicolumn{1}{c|}{\textbf{$\abs{E}$}} &
\multicolumn{1}{c}{\textbf{T}} &
\multicolumn{1}{c|}{\textbf{\#B}} &
\multicolumn{1}{c}{\textbf{T}} &
\multicolumn{1}{c}{\textbf{\#B}} &
\multicolumn{1}{c|}{\textbf{\#C}} &
\multicolumn{1}{c}{\textbf{T}} &
\multicolumn{1}{c|}{\textbf{\#B}} &
\multicolumn{1}{c}{\textbf{T}} &
\multicolumn{1}{c}{\textbf{\#A}} \\
\midrule
afro-29 & 552 & 2,392 & 0.4 & 66 & 0.3 & 0 & 55 & 0.2 & 827 & 0.1 & 56,095 \\
afro-45 & 160 & 852 & 0.2 & 55 & 0.4 & 3 & 20 & 0.0 & 402 & 0.1 & 99,765 \\
afro-43 & 126 & 679 & 0.2 & 32 & 0.3 & 41 & 25 & 0.0 & 240 & 0.1 & 102,609 \\
afro-39 & 288 & 1,528 & 0.3 & 58 & 1.3 & 0 & 34 & 0.0 & 876 & 0.4 & 281,403 \\
afro-40 & 272 & 1,376 & 0.4 & 80 & 0.2 & 0 & 47 & 0.0 & 856 & 0.5 & 333,793 \\
afro-28 & 334 & 1,875 & 0.5 & 72 & 0.9 & 22 & 51 & 0.1 & 2,768 & 0.6 & 464,272 \\
afro-38 & 342 & 1,895 & 0.3 & 51 & 0.2 & 0 & 80 & 0.3 & 4,776 & 0.9 & 631,053 \\
afro-14 & 250 & 1,175 & 0.2 & 50 & 1.0 & 9 & 40 & 0.0 & 326 & 2.2 & 1,707,228 \\
afro-19 & 382 & 1,481 & 0.3 & 53 & 1.0 & 0 & 34 & 0.1 & 773 & 1.9 & 1,803,293 \\
afro-24 & 516 & 2,474 & 0.3 & 42 & 0.3 & 0 & 69 & 0.5 & 4,372 & 5.3 & 1,998,636 \\
afro-17 & 302 & 1,417 & 0.9 & 197 & 1.1 & 87 & 62 & 0.3 & 7,046 & 3.1 & 2,342,879 \\
afro-42 & 472 & 2,456 & 0.6 & 114 & 2.0 & 45 & 76 & 0.2 & 2,014 & 41.4 & 22,588,100 \\
afro-32 & 286 & 1,643 & 0.8 & 469 & 1.6 & 133 & 56 & 0.5 & 12,151 & 49.1 & 29,512,013 \\
afro-41 & 592 & 3,536 & 1.0 & 294 & 2.6 & 37 & 78 & 1.5 & 9,334 & 128.2 & 55,758,998 \\
\bottomrule
\end{tabular*}
\end{table}
}

{
\begin{table}[p!]
\center
\caption{Comparison of branching rules.}
\label{tbl:branching_rule}
\fontsize{8pt}{0pt}\selectfont
\begin{tabular*}{\columnwidth}{@{\extracolsep{\fill}}l|rr|rr|rr}
\toprule
& \multicolumn{2}{c|}{\textbf{B0}}
& \multicolumn{2}{c|}{\textbf{B1}}
& \multicolumn{2}{c}{\textbf{B2}} \\
\multicolumn{1}{c|}{\textbf{Instance}} &
\multicolumn{1}{c}{\textbf{T}} & \multicolumn{1}{c|}{\textbf{\#B}} &
\multicolumn{1}{c}{\textbf{T}} & \multicolumn{1}{c|}{\textbf{\#B}} &
\multicolumn{1}{c}{\textbf{T}} & \multicolumn{1}{c}{\textbf{\#B}} \\
\midrule
petster-dog & 6.0 & 3 & 5.8 & 14 & 5.8 & 3 \\
hudong-internal & 2.3 & 19 & 1.8 & 9 & 1.9 & 5 \\
baidu-related & -- & -- & -- & -- & 2.1 & 8 \\
web-Google & 1.3 & 5 & 1.5 & 17 & 1.4 & 10 \\
soc-LiveJournal1 & 11.3 & 314 & 10.8 & 195 & 11.5 & 36 \\
libimseti & -- & -- & -- & -- & 1,642.8 & 472 \\
in-2004 & 37.2 & 30,344 & 28.4 & 21,377 & 3.5 & 668 \\
web-NotreDame & -- & -- & 687.8 & 356,138 & 13.4 & 4,266 \\
web-BerkStan & -- & -- & -- & -- & 142.3 & 42,270 \\
as-Skitter & -- & -- & -- & -- & 2,769.8 & 2,123,545 \\
web-Stanford & -- & -- & -- & -- & 67,270.3 & 55,865,269 \\
\bottomrule
\end{tabular*}
\end{table}
}

{
\begin{table}[p!]
\center
\caption{Comparison of reduction rules.}
\label{tbl:reduction_rules}
\fontsize{8pt}{0pt}\selectfont
\begin{tabular*}{\columnwidth}{@{\extracolsep{\fill}}l|rr|rr|rr|rr|rr}
\toprule
& \multicolumn{2}{c|}{\textbf{R0}}
& \multicolumn{2}{c|}{\textbf{R1}}
& \multicolumn{2}{c|}{\textbf{R2}}
& \multicolumn{2}{c|}{\textbf{R3}}
& \multicolumn{2}{c}{\textbf{R4}} \\
\multicolumn{1}{c|}{\textbf{Instance}} &
\multicolumn{1}{c}{\textbf{T}} & \multicolumn{1}{c|}{\textbf{\#B}} &
\multicolumn{1}{c}{\textbf{T}} & \multicolumn{1}{c|}{\textbf{\#B}} &
\multicolumn{1}{c}{\textbf{T}} & \multicolumn{1}{c|}{\textbf{\#B}} &
\multicolumn{1}{c}{\textbf{T}} & \multicolumn{1}{c|}{\textbf{\#B}} &
\multicolumn{1}{c}{\textbf{T}} & \multicolumn{1}{c}{\textbf{\#B}} \\
\midrule
petster-dog & -- & -- & -- & -- & 4,724.6 & 5,557,005 & 8.5 & 4 & 5.8 & 3 \\
hudong-internal & -- & -- & 24.5 & 185 & 2.0 & 8 & 2.0 & 5 & 1.9 & 5 \\
baidu-related & -- & -- & -- & -- & 1.9 & 152 & 1.9 & 8 & 2.1 & 8 \\
web-Google & -- & -- & 1.5 & 602 & 1.5 & 165 & 1.0 & 10 & 1.4 & 10 \\
soc-LiveJournal1 & -- & -- & -- & -- & 45.0 & 7,500 & 9.6 & 33 & 11.5 & 36 \\
libimseti & -- & -- & -- & -- & 837.0 & 476 & 1,371.5 & 472 & 1,642.8 & 472 \\
in-2004 & -- & -- & -- & -- & 28.0 & 30,824 & 4.6 & 1,504 & 3.5 & 668 \\
web-NotreDame & -- & -- & -- & -- & 747.7 & 1,088,096 & 29.3 & 16,563 & 13.4 & 4,266 \\
web-BerkStan & -- & -- & -- & -- & -- & -- & 7,986.2 & 3,898,313 & 142.3 & 42,270 \\
as-Skitter & -- & -- & -- & -- & 10,507.3 & 16,422,252 & 7,768.1 & 6,262,544 & 2,769.8 & 2,123,545 \\
web-Stanford & -- & -- & -- & -- & -- & -- & -- & -- & 67,270.3 & 55,865,269 \\
\bottomrule
\end{tabular*}
\end{table}
}

{
\tabcolsep = 1mm
\begin{table}[p!]
\center
\caption{Comparison of lower bounds.}
\label{tbl:lower_bounds}
\fontsize{8pt}{0pt}\selectfont
\begin{tabular*}{\columnwidth}{@{\extracolsep{\fill}}l|rr|rr|rr|rr|rr}
\toprule
& \multicolumn{2}{c|}{\textbf{L0}}
& \multicolumn{2}{c|}{\textbf{L1}}
& \multicolumn{2}{c|}{\textbf{L2}}
& \multicolumn{2}{c|}{\textbf{L3}}
& \multicolumn{2}{c}{\textbf{L4}} \\
\multicolumn{1}{c|}{\textbf{Instance}} &
\multicolumn{1}{c}{\textbf{T}} & \multicolumn{1}{c|}{\textbf{\#B}} &
\multicolumn{1}{c}{\textbf{T}} & \multicolumn{1}{c|}{\textbf{\#B}} &
\multicolumn{1}{c}{\textbf{T}} & \multicolumn{1}{c|}{\textbf{\#B}} &
\multicolumn{1}{c}{\textbf{T}} & \multicolumn{1}{c|}{\textbf{\#B}} &
\multicolumn{1}{c}{\textbf{T}} & \multicolumn{1}{c}{\textbf{\#B}} \\
\midrule
petster-dog & 6.1 & 122 & 6.2 & 15 & 6.0 & 7 & 5.4 & 3 & 5.8 & 3 \\
hudong-internal & 2.0 & 16 & 1.9 & 6 & 2.0 & 6 & 1.9 & 5 & 1.9 & 5 \\
baidu-related & 74.7 & 27,047 & 1.9 & 8 & 2.7 & 77 & 2.6 & 68 & 2.1 & 8 \\
web-Google & 1.1 & 47 & 1.0 & 10 & 1.3 & 30 & 1.0 & 30 & 1.4 & 10 \\
soc-LiveJournal1 & 12.6 & 180 & 10.6 & 37 & 12.0 & 122 & 11.0 & 117 & 11.5 & 36 \\
libimseti & -- & -- & -- & -- & 1,747.0 & 476 & 1,776.5 & 472 & 1,642.8 & 472 \\
in-2004 & 4.4 & 878 & 4.1 & 683 & 5.2 & 781 & 4.4 & 757 & 3.5 & 668 \\
web-NotreDame & 17.9 & 13,740 & 14.2 & 5,678 & 13.7 & 5,702 & 13.7 & 4,447 & 13.4 & 4,266 \\
web-BerkStan & 261.2 & 206,209 & 145.3 & 42,503 & 199.6 & 113,184 & 176.0 & 62,877 & 142.3 & 42,270 \\
as-Skitter & 4,963.5 & 6,973,582 & 2,629.6 & 2,153,280 & 4,165.8 & 4,873,893 & 3,968.9 & 4,083,355 & 2,769.8 & 2,123,545 \\
web-Stanford & -- & -- & -- & -- & -- & -- & 64,757.6 & 55,912,396 & 67,270.3 & 55,865,269 \\
\bottomrule
\end{tabular*}
\end{table}
}

\subsection{Analysis}
Finally, we examine the effect of
branching strategies, reduction rules and lower bounds.

\subsubsection{Branching Rules}
We compared the following three branching strategies.
\emph{B0} selects a vertex to branch on in a uniformly random manner,
\emph{B1} branches on a vertex with the minimum degree and
\emph{B2} chooses a vertex with the maximum degree.

Table~\ref{tbl:branching_rule} lists the results,
which show that selecting a vertex with the maximum degree (B2)
is significantly better than other strategies.
This matches the results of theoretical research.
Another interesting finding here is that
the minimum degree strategy (B1) performs better than the random strategy (B0).
This is because
our algorithm incorporates mirror branching (Section~\ref{sec:branching:mirror}),
which occurs more often when branching on vertices with small degrees.

\subsubsection{Reduction Rules}
To examine the effects of reduction rules,
we compare algorithms \emph{R0}--\emph{R4},
which use different sets of reduction rules.
\emph{R0} does not use any reduction rules other than connected component decomposition.
\emph{R1} uses the first three reduction rules: degree-1, dominance, and degree-2 folding
(Section~\ref{sec:reduction:exponential}).
In addition to the first three reduction rules, \emph{R2} uses the LP-based reduction rule (Section~\ref{sec:reduction:lp}).
\emph{R3} also adopts unconfined, twin, funnel, and desk (Section~\ref{sec:reduction:sparse}).
\emph{R4} uses all the reduction rules,
including the packing rule (Section~\ref{sec:reduction:packing}),
which is newly introduced in this paper.

Results are listed in Table~\ref{tbl:reduction_rules}.
We can observe the significant effect of reduction rules on the search space.
Indeed, without reduction rules, R0 cannot solve any problems.
On the other hand, we confirm that search space gets smaller and smaller
by introducing reduction rules
on instances such as web-Google, web-NotreDame and as-Skitter.
We can also see that
the number of problems that can be solved within the time limit
increases by adopting reduction rules.

\subsubsection{Lower Bounds}
\label{sec:experiments:lower_bounds}
Finally, we compare algorithms \emph{L0}--\emph{L4}
using different lower bounds.
\emph{L0} only uses the number vertices currently included to the vertex cover.
\emph{L1}, \emph{L2}, and \emph{L3}
use the clique cover (Section~\ref{sec:lower_bounds:clique}),
LP relaxation (Section~\ref{sec:lower_bounds:lp}), and
cycle cover (Section~\ref{sec:lower_bounds:cycle}), respectively.
\emph{L4} combines all these lower bounds.

Table~\ref{tbl:lower_bounds} describes the results.
It shows that the difference of lower bounds does not drastically affect the
results in comparison to the branching rules and the reduction rules.
As expected, the search space of L4 is the smallest among the five methods in all the instances.
Since L3 is an extension of L2, it works better than L2 in all the instances.
Although L1 works better than L3 in some instances, L3 works better in the other instances.

\section{Conclusion}\label{sec:conclusion}
We investigated the practical impact of theoretical research on branching and reduction rules.
Our experimental results indicated that,
as well as theoretical importance,
development of these techniques indeed leads to empirical efficiency.

\subsubsection*{Acknowledgements.}
The authors are supported by Grant-in-Aid for JSPS Fellows 256563 and 256487.

\bibliographystyle{abbrv}
\bibliography{main}

\newpage
\section*{Appendix}
\appendix
\section{Parameterized Complexity of Vertex Cover above Lower Bounds}\label{sec:aboveLB}
The previous theoretical research has shown that if the LP relaxation gives a lower bound that is close to the optimal
value, \textsc{Vertex Cover} can be efficiently solved in the context of
parameterized complexity~\cite{DBLP:journals/corr/abs-1203-0833}.
In our algorithm, we used two different lower bounds, clique cover and cycle cover, which can give a better lower
bound than LP relaxation.
In this section, we investigate the parameterized complexity of \textsc{Vertex Cover} above these lower bounds and show
that even if these lower bounds are very close to the optimal value, \textsc{Vertex Cover} can become very difficult.

\subsection{Vertex Cover above Clique Cover}
Let us define a parameterized problem, \textsc{Vertex Cover above Clique Cover}.
In this problem, we are given a graph $G$, a clique cover $\mathcal{C}$ of $G$, and a parameter $k$;
our objective is to find a vertex cover of size at most $|V|-|\mathcal{C}|+k$.
Here, $|V|-|\mathcal{C}|$ is the lower bound of the optimal solution obtained from the given clique
cover.
In contrast to LP lower bound,
we prove that this parameterized problem is parameterized NP-hard; i.e., even if the difference between the lower
bound obtained from the clique cover and the optimal value is constant, \textsc{Vertex Cover} is still NP-hard.
\begin{theorem}\label{thm:above_clique}
\textsc{Vertex Cover above Clique Cover} is parameterized NP-hard.
\end{theorem}

\begin{proof}
We prove the theorem by a reduction from \textsc{3-SAT}.
Let $(X,\mathcal{F})$ be an instance of \textsc{3-SAT},
where $X=\{x_1,x_2,\ldots,x_n\}$ is a set of variables and $\mathcal{F}=\{F_1,F_2,\ldots,F_m\}$ is a set of 3-clauses on
$X$.
The negation of a variable $x$ is denoted by $\bar{x}$.
We write each clause $F_i$ as $F_i=(l_{i,1}\vee l_{i,2}\vee l_{i,3})$, where $l_{i,j}$ is a literal of $X$, i.e.,
$l_{i,j}=x$ or $\bar{x}$ for some $x\in X$.

We reduce the instance of \textsc{3-SAT} to an instance of \textsc{Vertex Cover above Clique Cover} with a
parameter $k=0$ as follows.
For each variable $x_i\in X$, we create two vertices $v_i$ and $\bar{v}_i$ and connect them by an edge.
Let $f$ be a function that maps a literal $x_i$ to the vertex $v_i$ and a literal $\bar{x}_i$ to the vertex $\bar{v}_i$.
For each clause $F_i=(l_{i,1}\vee l_{i,2}\vee l_{i,3})\in\mathcal{F}$, we create three vertices $u_{i,1}$, $u_{i,2}$,
and $u_{i,3}$, and connect them to form a triangle.
Then, for each $j=1,2,3$, we connect $u_{i,j}$ to $f(l_{i,j})$.
Finally, we construct a clique cover $\mathcal{C}$ by taking a clique $\{v_i,\bar{v}_i\}$ of size two from each
variable $x_i\in X$, and a clique $\{u_{i,1},u_{i,2},u_{i,3}\}$ of size three from each clause $F_i\in\mathcal{F}$.
The number of the vertices is $2n+3m$, and the size of this clique cover is $n+m$.
Thus, the lower bound obtained from the clique cover is $n+2m$.

Finally, we prove that, if and only if the instance of \textsc{3-SAT} is satisfiable, the reduced graph has a vertex
cover of size $n+2m$.

$(\Rightarrow)$
We construct a vertex cover $C$ as follows.
Let $\pi$ be a truth assignment that satisfies all the clauses.
For each variable $x_i\in X$, if $\pi(x_i)$ is true, we include $v_i$ to $C$; otherwise, we include
$\bar{v}_i$ to $C$.
This covers an edge between $v_i$ and $\bar{v}_i$.
For each clause $F_i\in\mathcal{F}$, we choose a literal $l_{i,j}$ such that $\pi(l_{i,j})$ is true.
Since $\pi$ is a satisfying assignment, we can always choose such a literal.
Then, we include the two vertices other than $u_{i,j}$ from the triangle $\{u_{i,1},u_{i,2},u_{i,3}\}$ to $C$.
These cover the edges on the triangle.
Moreover, for each $j=1,2,3$, if $\pi(l_{i,j})$ is true, $f(l_{i,j})$ is in $C$; otherwise, $u_{i,j}$ is in $C$.
Therefore, the edge between $u_{i,j}$ and $f(l_{i,j})$ is also covered.
Thus, all the edges are covered by $C$; i.e., $C$ is a vertex cover.
Apparently, the size of $C$ is $n+2m$.

$(\Leftarrow)$
We construct a satisfying assignment $\pi$ as follows.
Let $C$ be a vertex cover of size $n+2m$.
Since the lower bound obtained from the clique cover $\mathcal{C}$ is also $n+2m$, this implies that for each clique
$C_i\in\mathcal{C}$, $C$ contains exactly $|C_i|-1$ vertices from $C_i$.
Therefore, for each variable $x_i\in X$, $C$ contains exactly one of $v_i$ and $\bar{v}_i$.
If $v_i$ is contained in $C$, we assign $\pi(x_i)$ to true; otherwise, we assign $\pi(x_i)$ to false.
Now, we show that this assignment $\pi$ satisfies all the clauses.
For each clause $F_i\in\mathcal{F}$, since $C$ contains exactly two vertices from the triangle
$\{u_{i,1},u_{i,2},u_{i,3}\}$, exactly one vertex $u_{i,j}$ of them is not contained in $C$.
Since $C$ is a vertex cover, its adjacent vertex $f(l_{i,j})$ is contained in $C$.
Thus, $\pi(l_{i,j})$ is true, and therefore, $F_i$ is satisfied by $\pi$.
\end{proof}

\subsection{Vertex Cover above Cycle Cover}

Let us define another parameterized problem \textsc{Vertex Cover above Cycle Cover}.
Similar to the previous problem, we are given a graph $G$, a cycle cover $\mathcal{C}$ of $G$,
and a parameter $k$, and our objective is to find a vertex cover of size at most
$\sum_{C\in\mathcal{C}}\left\lceil\frac{|C|}{2}\right\rceil+k$.
Here, $\sum_{C\in\mathcal{C}}\left\lceil\frac{|C|}{2}\right\rceil$ is the lower bound of the optimal solution obtained
from the given cycle cover.
Similar to \textsc{Vertex Cover above Clique Cover}, \textsc{Vertex Cover above Cycle Cover} also becomes parameterized
NP-hard.

\begin{theorem}\label{thm:above_cycle}
\textsc{Vertex Cover above Cycle Cover} is parameterized NP-hard.
\end{theorem}

\begin{proof}
The proof is almost the same.
The size of each clique $C$ in the clique cover $\mathcal{C}$ used in the proof of Theorem~\ref{thm:above_clique} is two
or three.
Thus, the clique cover $\mathcal{C}$ is also a cycle cover.
Moreover, when $|C|=2$ or $3$, $|C|-1$ equals to $\left\lceil\frac{|C|}{2}\right\rceil$.
Therefore, the lower bound obtained by considering $\mathcal{C}$ as a cycle cover exactly matches the lower bound
obtained by considering $\mathcal{C}$ as a clique cover.
Thus, we can use the same argument.
\end{proof}

\end{document}